\documentclass[journal]{IEEEtran}
\ifCLASSINFOpdf
\else
\fi
\usepackage{setspace}

\usepackage{cite}
\usepackage{amsmath,amssymb,amsfonts,amsthm}
\usepackage{algorithm}
\usepackage{algorithmic}
\usepackage{graphicx}
\usepackage{comment}
\usepackage{textcomp}
\usepackage{subcaption}
\usepackage{color}
\usepackage{stfloats}
\newtheorem{lemma}{\quad \textit{Lemma}}
\usepackage{hyperref}

\def\BibTeX{{\rm B\kern-.05em{\sc i\kern-.025em b}\kern-.08em
    T\kern-.1667em\lower.7ex\hbox{E}\kern-.125emX}}
\usepackage{bm}      

\begin{document}

\title{Generative AI Empowered Semantic Feature Multiple Access (SFMA) Over Wireless Networks }

\author{Jiaxiang Wang, Yinchao Yang,
Zhaohui Yang, Chongwen Huang, Mingzhe Chen, Zhaoyang Zhang, 
and Mohammad Shikh-Bahaei,
\IEEEmembership{Senior Member, IEEE}
\thanks{Jiaxiang Wang, Yinchao Yang, and Mohammad Shikh-Bahaei are with the Department of Engineering, King's College London, London, UK. (emails: jiaxiang.wang@kcl.ac.uk;yinchao.yang@kcl.ac.uk; m.sbahaei@kcl.ac.uk)}
\thanks{Zhaohui Yang, Chongwen Huang, and Zhangyang Zhang are with the College of Information Science and Electronic Engineering, Zhejiang University, Hangzhou, Zhejiang 310027, China, and Zhejiang Provincial Key Lab of Information Processing, Communication and Networking (IPCAN), Hangzhou, Zhejiang, 310007, China. (email: yang\_zhaohui@zju.edu.cn; chongwenhuang@zju.edu.cn; ning\_ming@zju.edu.cn) }
\thanks{Mingzhe Chen is with the Department of Electrical and Computer Engineering and Institute for Data Science and Computing, University of Miami, Coral Gables, FL, 33146, USA (e-mail: mingzhe.chen@miami.edu).
}
}

\maketitle
\begin{abstract}
This paper investigates a novel generative artificial intelligence (GAI) empowered multi-user semantic communication system called semantic feature multiple access (SFMA) for video transmission, which comprises a base station (BS) and paired users. The BS generates and combines semantic information of several frames simultaneously requested by paired users into a signal. Users recover their frames from this combined signal and input the recovered frames into a GAI-based video frame interpolation model to generate the intermediate frame.
To optimize transmission rates and temporal gaps between simultaneously transmitted frames, we formulate an optimization problem to maximize the system sum rate while minimizing temporal gaps. We observe that the standard signal-to-interference-plus-noise ratio (SINR) equation does not accurately capture the performance of our semantic communication system. Therefore, we introduce a weight parameter into the SINR equation to better represent the system's performance. Due to the complexity introduced by the weight parameter's dependence on transmit power, we propose a three-step solution.
First, we develop a user pairing algorithm that pairs two users with the highest preference value, a weighted combination of semantic transmission rate and temporal gap. Second, we optimize inter-group power allocation by formulating an optimization problem that allocates proper transmit power across all user groups to maximize system sum rates while satisfying each user's minimum rate requirement. Third, we address intra-group power allocation to enhance the performance of each user.
Simulation results demonstrate that our method improves transmission rates by up to $24.8\%$, $45.8\%$, and $66.1\%$ compared to fixed-power non-orthogonal multiple access (F-NOMA), orthogonal joint source-channel Coding (O-JSCC), and orthogonal frequency division multiple access (OFDMA) schemes, respectively.

\end{abstract}

\section{Introduction}
With the development of edge devices such as computing hardware, human intelligence-based wireless applications (e.g., smartphones, virtual reality glasses, robots, etc.) have emerged. However, current communication systems may not be able to support such emerging applications \cite{ref56, ref66}. Semantic communications that enable devices to exploit the knowledge of the transmitter and receivers to extract and transmit data meaning instead of the entire source data seems a promising solution. However, realizing semantic communications faces several challenges such as efficient semantic information extraction, robustness to errors, interoperability across devices, and managing semantic interference in multi-user environments.

Several existing works \cite{ref62, ref63, ref55, ref22, ref23, ref24, ref34, ref35, ref64, ref40, ref65} have studied the design of semantic communication for data transmission across different modalities.
In particular, the authors in \cite{ref62} proposed an end-to-end semantic communication system that leverages the transformer to enhance the learning ability of the meaning of sentences. The work in \cite{ref63} extended this work to task-oriented multi-user semantic communications for transmitting multimodal data. The authors in \cite{ref55} designed a semantic communication framework for textual data transmission by using knowledge graphs. The authors in \cite{ref22} developed a semantic communication system for speech transmission, which leverages attention mechanisms to recover key speech features for robust speech data transmission. The authors in \cite{ref23} introduced a deep neural network (DNN) based joint source-channel coding (JSCC) scheme for image data transmission. The authors in \cite{ref24} studied the combination of data compression and channel coding to improve video transmission performance with various channel conditions. Meanwhile, a number of prior works \cite{ref53, ref57, ref58} have also studied the optimization of semantic communication performance over wireless networks. 
However, the methods designed in \cite{ref62, ref55, ref22, ref23, ref24, ref53, ref57, ref58} may not be applied for dynamic scenarios, where the transmission environment, heterogeneous data sources, and multiple task requirements vary over time. To apply semantic communications for dynamic scenarios, generative artificial intelligence (GAI) seems a promising technique due to its ability for information fusion and the corresponding network optimization. \cite{ref59}.

Currently, several works \cite{ref34, ref35, ref64, ref40, ref65} have leveraged GAI for semantic communications. In particular, the authors in \cite{ref34} proposed a channel denoising diffusion model (CDDM) to learn the distribution of signals and assist in removing channel noise for a JSCC image transmission semantic communication system. The work in \cite{ref35} used the diffusion model-based channel enhancer (DMCE) to enhance the estimation of the channel state, thus improving the recovery of semantic information. The authors in \cite{ref64} proposed a hybrid JSCC scheme to combine the signals transmitted by the digital communication scheme with the JSCC signal generated by the diffusion model to achieve a high reconstruction performance.
The authors in \cite{ref40} utilized a trained style-based generative adversarial network (StyleGAN) model for the extraction and reconstruction of semantic information in JSCC, improving image transmission efficiency while ensuring the quality of the received images. 
However, none of these works considered using multiple access (MA) techniques to improve semantic communication performance. In fact, by using MA techniques, devices can reuse spectrum from data meaning aspect significantly thus improving semantic communication performance. Although a few works \cite{ref60, ref13, ref61, ref14} have investigated the use of MA for semantic communications, they did not capture the complex interactions and interference between signals at the semantic level, and their impacts on semantic communication performance optimization.


\begin{figure}[t]
\centerline{\includegraphics[width=0.45\textwidth]{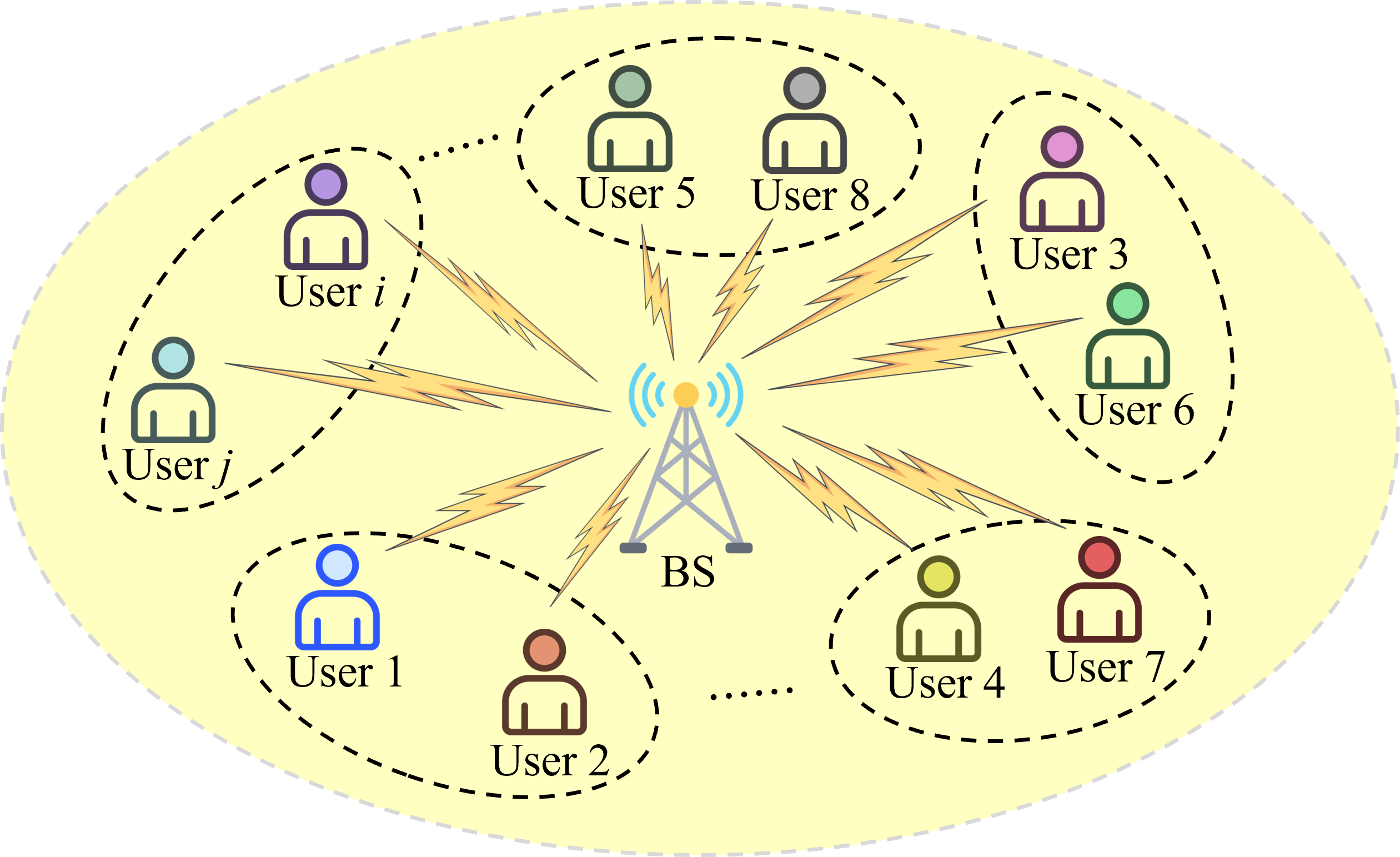}}
\captionsetup{font={small}}
\caption{System model of a downlink SFMA single cell network.}
\label{fig: SFMA}
\end{figure}

The main contribution of our work is a novel multi-user access semantic communication framework named semantic feature multiple access (SFMA), shown in Fig. \ref{fig: SFMA}. 
SFMA is a kind of non-orthogonal multiple access, where some users share the same resource block to transmit information. To decrease inter-user interference, each user first extracts its semantic feature, which will be transmitted over the shared resource block with other users. Due to the different semantic information spaces of different users, the inter-user interference will be effectively controlled by using the semantic decoder at the receiver side. 
The specific contributions are summarized as follows:
\begin{itemize} 
    \item We propose a novel multi-user access semantic communication system called SFMA, and apply it to video transmission. Our considered semantic system consists of one BS and multiple users. The users are paired into several groups and each group consists of two users. The BS first generates the semantic information of the source frames requested by users via semantic encoders. Then, the BS combines the semantic information that will be transmitted to the two users in a group into a single transmitted signal using their semantic features. The users will extract their frames from the combined semantic information. 
    Compared to standard semantic communication systems, our designed system enables the BS to simultaneously transmit data to the users within each group thus significantly improving spatial efficiency.
    
    \item To jointly optimize their transmission rates and the temporal gaps between the simultaneously transmitted video frames, we formulate an optimization problem aiming to maximize the system's transmission rate while minimizing the temporal gap between frames required by the paired users. 
    To solve this problem, we first use simulations to verify that the standard equation of the SINR ratio cannot capture the performance of our SFMA system since it cannot capture data meaning transmission performance. We then introduce a weight parameter into the standard SINR equation to accurately capture the performance of our designed semantic system. 
    
    \item Since the weight parameter in the proposed SINR equation depends on the transmit power of the users, which will significantly increase the complexity of solving the problem, we decompose the formulated problem into three sub-problems: 1) user pairing strategy, 2) intra-group power allocation, and 3) inter-group power allocation. 
    To solve the first sub-problem, we introduce a user pairing algorithm to maximize the sum transmission rate of the SFMA system while minimizing the temporal gap between simultaneously transmitted frames. Then, we optimize the inter-group power allocation in the SFMA system aiming to maximize the sum data rate of all users while satisfying each user's minimum data rate requirement. Finally, we optimize the intra-group power allocation based on the optimized inter-group power allocation strategy to further improve each user's semantic transmission rate within each group.
    
\end{itemize}
Simulation results show that our proposed method can improve the transmission rate by up to $24.8\%$, $45.8\%$, and $66.1\%$ respectively, compared to the fixed-power non-orthogonal multiple access (F-NOMA), orthogonal joint source-channel Coding (O-JSCC), and orthogonal frequency division multiple access (OFDMA) schemes.

\begin{figure*}[t]
\centerline{\includegraphics[width=0.9\textwidth]{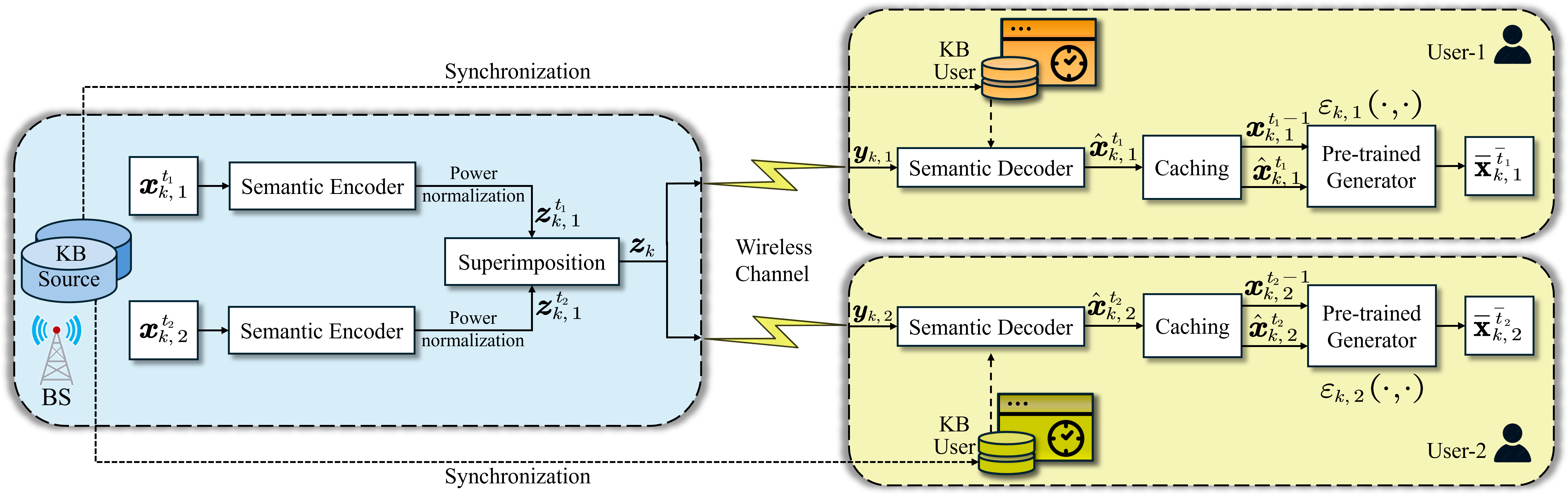}}
\captionsetup{font={small}}
\caption{The framework of GAI-enabled semantic video transmission in the SFMA system.}
\label{fig: system model}
\end{figure*}

The rest of the paper is organized as follows: Section \ref{sec: system model} elaborates on the system model of the downlink SFMA framework. Section \ref{sec: problem} presents the semantic interference factor and the overall problem formulation. Section \ref{sec: solutions} illustrates the user pairing strategy, inter-group power allocation algorithm for multiple groups, and intra-group power allocation algorithm within each group in the SFMA system. Section \ref{sec: results} provides experimental results and analysis, demonstrating the effectiveness of the proposed methods. Finally, Section \ref{sec: conclusion} concludes the paper and discusses potential future research directions.

\section{System Model}
\label{sec: system model}


We consider a downlink semantic communication system that consists of one base station (BS) and $M$ single-antenna users as shown in Fig. \ref{fig: SFMA}. The $M$ users are divided into $K=M/2$ groups  (we consider the case that $M$ is an even number) such that each group consists of two users, as shown in Fig. \ref{fig: system model}. The specific user grouping strategy will be described in Section \ref{sec: user pairing}. Each group is assumed to have users 1 and 2. Let $\mathcal{M}=\left\{ 1,\cdots, M \right\} $ and $\mathcal{K}=\left\{ 1,\cdots, K \right\} $ denote the set of users and groups, respectively. Note that we implement orthogonal multiple access (OMA) among different user groups. The BS will first generate semantic information of source images requested by the users via semantic encoders.
Then, the BS will combine the semantic information that will be transmitted to the two users in group $k$ into a single transmitted signal using their semantic features. 
Upon receiving the compressed semantic information, two users reconstruct the original frames and utilize GAI-based video interpolation models to generate intermediate frames. During the transmission, we assume that the raw frames $\boldsymbol{x}_{k,1}^{t_1}$ and $\boldsymbol{x}_{k,2}^{t_2}$ at time $t_1$ and $t_2$ required by two users in group $k$, respectively are expected to be transmitted. The semantic information of user 1 in group $k$ is $\boldsymbol{\tilde{z}}^{t_1}_{k,1}\left( \boldsymbol{x}^{t_1}_{k,1};\boldsymbol{\Theta }_{k,1} \right),$ while the semantic information of user 2 is $ \boldsymbol{\tilde{z}}^{t_2}_{k,2}\left( \boldsymbol{x}^{t_2}_{k,2};\boldsymbol{\Theta }_{k,2} \right)$ with $\boldsymbol{\Theta }_{k,i}$ being the learnable parameters of the semantic encoder of user $i$ in group $k$. 
Let $\boldsymbol{z}^{t_1}_{k, 1}$ and $\boldsymbol{z}^{t_2}_{k, 2}$ denote the normalized versions of the signals $\boldsymbol{\tilde{z}}^{t_1}_{k,1}$ and $\boldsymbol{\tilde{z}}^{t_2}_{k,2}$, respectively. 
Then, the signal that combines $\boldsymbol{z}^{t_1}_{k, 1}$ and $\boldsymbol{z}^{t_2}_{k, 2}$ at the BS is 
\begin{equation}
    \boldsymbol{z}_{k}=\sqrt{p_{k ,1}}\boldsymbol{z}^{t_1}_{k, 1}+\sqrt{p_{k, 2}}\boldsymbol{z}^{t_2}_{k, 2},
\label{eq: signal SC}
\end{equation}
where $p_{k,1}$ and $p_{k,2}$ are the allocated power used for the semantic information transmission of user 1 and user 2 respectively in group $k$, $\eta _{k, 1}, \eta _{k, 2}$ denote the power allocation factor for user 1 and user 2 in group $k$ with $\sqrt{\frac{p_{k,1}}{p_{k,2}}} = \frac{\eta_{k,1}}{\eta_{k,2}}$, respectively. The BS then sends the superimposed signal $\boldsymbol{z}_{k}$ through a wireless noisy channel. Let $h_{k, i}$ be the channel coefficient between the BS and user $i$ in group $k$, the received signal of user $i$ in group $k$ is denoted as
\begin{equation}
    \boldsymbol{y}_{k, i}=h_{k, i}\boldsymbol{z}_{k}+n_{k, i}, i\in \left\{ 1,2 \right\},
\label{eq: y_ki}
\end{equation}
where $n_{k, 1}$ and $n_{k, 2}$ are independent and identically distributed (i.i.d.) Gaussian variables with variances $\sigma _{k, 1}^{2}$ and $\sigma _{k, 2}^{2}$, respectively, i.e., $n_{k,i}\sim \mathcal{N}\left( 0, \sigma _{k, i}^{2}\right) , i\in \left\{ 1,2 \right\}$. The signal-to-noise ratio (SNR) of user $i$ in group $k$ is defined as
\begin{equation}
    \text{SNR}_{k,i} = 10 \log_{10} \left( \frac{p_{k,i}}{\sigma_{k,i}^2} \right)\text{dB}. 
\end{equation}
Each user $i$ reconstructs frame $\boldsymbol{x}^{t_i}_{k, i}$ from its received signal $\boldsymbol{y}_{k, i}$. The reconstructed frame is represented by $\boldsymbol{\hat{x}}^{t_i}_{k, i}\left( \boldsymbol{y}_{k, i} ;\boldsymbol{\Phi }_{k, i} \right), i\in \left\{ 1,2 \right\}$, where $\boldsymbol{\Phi }_{k, i}$ is the learnable parameter of the semantic decoder of user $i$ in group $k$.


Once the raw video frames are reconstructed, they are stored in the caching unit. The caching unit adjusts the order of the video frames based on the sequence information provided by the knowledge base (KB) on the user side. After the adjustment, two consecutive frames are simultaneously transmitted to a GAI-based pre-trained generator to produce intermediate frames.
The input of the pre-trained generators is the recovered frame $\boldsymbol{\hat{x}}^{t_i}_{k, i}$ and its preceding frame $\boldsymbol{x}_{k,i}^{t_i - 1}$.
The pre-trained generators are GAI-based video frame interpolation models which generate intermediate frames $\boldsymbol{\bar{x}}_{k, i}^{\bar{t}_i}$ for user $i$, resulting in smoother video playback, where $\bar{t}_i\in \left( t_i - 1,t_i \right), i \in \left\{1, 2\right\}$. The process can be denoted as
\begin{equation}
\boldsymbol{\bar{x}}_{k,i}^{\bar{t}_i}=\boldsymbol{\varepsilon }_{k,i}\left( \boldsymbol{x}_{k,i}^{t_i - 1},\boldsymbol{\hat{x}}_{k,i}^{t_i} \right), i \in \left\{1, 2\right\},
\end{equation}
where $\boldsymbol{\varepsilon }_{k, i}\left( \cdot, \cdot \right)$ denotes the GAI-enabled video frame interpolation models which are illustrated in Section \ref{sec: VFI}.

\subsection{Semantic Encoder and Decoder}
Next, we introduce the models of our designed semantic encoder and decoder.

\begin{figure*}[t]
\centerline{\includegraphics[width=\textwidth]{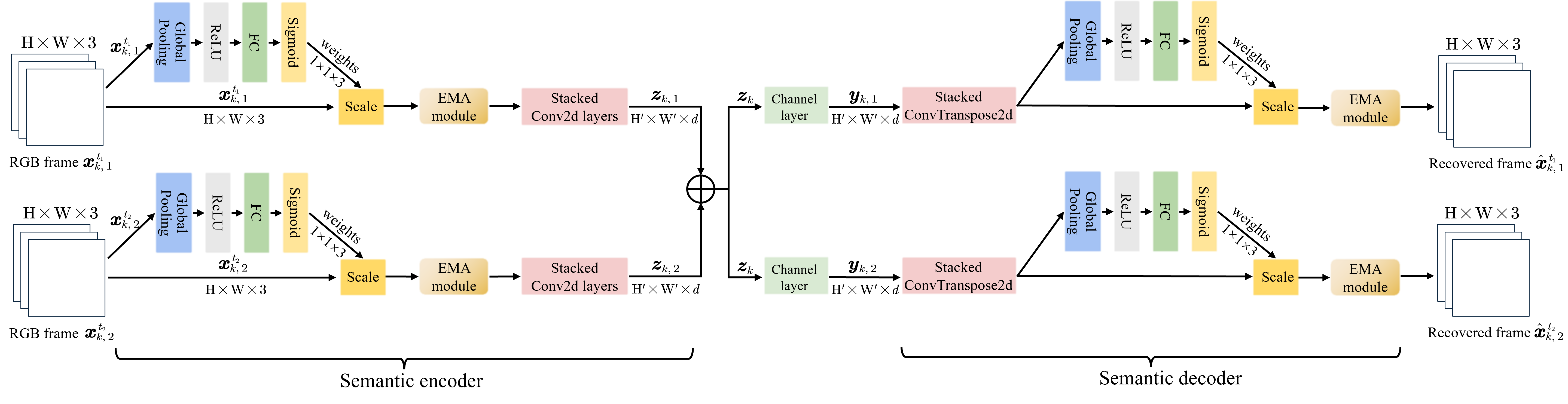}}
\captionsetup{font={small}}
\caption{The architecture of the Attention-based JSCC network.}
\label{fig: JSCC}
\end{figure*}

As shown in Fig. \ref{fig: JSCC}, the designed encoder consists of an input layer, attention modules, stacked convolutional layers, and a channel layer. On the encoder side, dimensionality reduction is essential to match the channel bandwidth and compress the transmitted data. However, directly reducing the dimensionality of images can significantly impair the dependencies among pixels and channels of images. Given that variations between consecutive frames are generally small, we integrate an attention mechanism to enhance the interdependencies among image pixels, thereby minimizing the mutual interference of latent feature vectors generated from two frames. 
The specific structure of the encoder is illustrated as follows:

\begin{itemize}
    \item Input layer: The input of the encoder is the raw frames required by users in one group. The frames can be considered as RGB images which have three channels R, G, and B.
    \item Attention modules: The attention modules are composed of two blocks. First, we use the Squeeze-and-Excitation networks \cite{ref44}. The weights for these channels are computed by calculating the global average values of each channel to evaluate the significance of the R, G, and B channels of the video frames. A single dense layer is employed to adjust the importance of each channel adaptively. The original features are subsequently scaled to the derived weights. Second, we leverage the efficient multi-scale attention (EMA) module \cite{ref43} to enhance feature representation capabilities, where different scales of convolutions are implemented and the output features of the two parallel branches are further aggregated by a cross-dimension interaction for capturing pixel-level pairwise relationship.
    \item Convolutional layers: several convolutional layers are utilized to compress the latent feature vector. The final convolutional layer in the encoder aims to achieve the desired \textit{compression ratio}.
    We call the input frame dimension $n= H \times W \times 3$ as the \textit{source bandwidth}, where $H$ and $W$ denote the height and width of the frame and the third dimension denotes R, G, and B channels. The \textit{compression ratio} which characterizes the available channel resources is defined as $\gamma=b/n$, and \textit{b} refers to the \textit{channel bandwidth}. 
    \item Channel layers: The AWGN channel is considered as one untrainable layer at the end of semantic encoders.
\end{itemize}
The semantic decoder includes convolutional layers and the EMA module. The semantic encoder and decoder are designed to be jointly trained to minimize the distortion of the original frames $\boldsymbol{x}_{k, 1}^{t_1}$ and $\boldsymbol{x}_{k, 2}^{t_2}$, and their reconstructions $\boldsymbol{\hat{x}}_{k, 1}^{t_1}$ and $\boldsymbol{\hat{x}}_{k, 2}^{t_2}$. We use mean square error (MSE) to measure the distortion, and the loss function can be expressed as
\begin{equation}
    \mathcal{L}=\sum_{\mathbf{\varPsi}}{\frac{1}{2}\text{MSE}\left( \boldsymbol{x}_{k, 1}^{t_1},\boldsymbol{\hat{x}}_{k, 1}^{t_1} \right) +\frac{1}{2}\text{MSE}\left( \boldsymbol{x}_{k, 2}^{t_2},\boldsymbol{\hat{x}}_{k, 2}^{t_2} \right)},
\end{equation}
where $\mathbf{\varPsi} = \left( \boldsymbol{x}_{k,1}^{t_1},\boldsymbol{x}_{k, 2}^{t_2} \right)
	\in \mathcal{D}_{train}$, $\mathcal{D}_{train}$ is the training dataset, and $\text{MSE}\left( \boldsymbol{x},\boldsymbol{\hat{x}} \right) =\frac{1}{n}\lVert \boldsymbol{x}-\boldsymbol{\hat{x}} \rVert ^2$.

\subsection{GAI-based Video Frame Interpolation}
\label{sec: VFI}

Next, we introduce a GAI-based video frame interpolation model to generate the intermediate frame.
In particular, we propose to use the Cross-Attention Transformer for Video Interpolation (TAIN) model in \cite{ref42}. Next, we explain how it works for users $i\in \left\{ 1,2 \right\}$ within a group. 
The input of the TAIN model is $\boldsymbol{x}_{k,i}^{t_i - 1}$ and $\boldsymbol{\hat{x}}_{k, i}^{t_i}$, and the output is the predicted intermediate frame $\mathbf{\bar{x}}_{k, i}^{\bar{t}_i}$.
 TAIN mainly consists of a Cross Similarity (CS) module and an Image Attention (IA) module. The CS module is a novel vision transformer module, which globally aggregates features from two input frames, $\boldsymbol{x}_{k, i}^{t_i - 1}$ and $\boldsymbol{\hat{x}}_{k, i}^{t_i}$, that is similar in appearance to those in the current prediction $\boldsymbol{\bar{x}}_{k, i}^{\bar{t}_i}$ of the raw intermediate frame $\boldsymbol{x}_{k, i}^{\bar{t}_i}$. 
To handle occlusions in the interpolated features within the aggregated CS features, the IA module is used to prioritize CS features in order to enhance the overall interpolation quality. 
To train the TAIN model, we use a loss function that includes the difference between the predicted intermediate frame $\mathbf{\bar{x}}_{k, i}^{\bar{t}_i}$ and the raw intermediate frames $\mathbf{x}_{k, i}^{\bar{t}_i}$. To maintain sharpness and preserve the edges of the generated interpolated frames, thereby enhancing the overall quality of the video interpolation, the $L_1$ loss on the gradient difference is incorporated in the loss function $\mathcal{L}_{\text{TAIN}, i}$ of the TAIN model, which is denoted by
\begin{equation}
    \mathcal{L}_{\text{TAIN}, i}=\lVert \mathbf{\bar{x}}_{k, i}^{\bar{t}_i}-\mathbf{x}_{k, i}^{\bar{t}_i} \rVert _1+\kappa \lVert \nabla \mathbf{\bar{x}}_{k, i}^{\bar{t}_i}-\nabla \mathbf{x}_{k, i}^{\bar{t}_i} \rVert _1, i\in \left\{ 1,2 \right\},
\end{equation}
where $\lVert \mathbf{\bar{x}}_{k, i}^{\bar{t}_i}-\mathbf{x}_{k, i}^{\bar{t}_i} \rVert _1$ is the L1-norm of the difference between the predicted intermediate frame $\mathbf{\bar{x}}_{k, i}^{\bar{t}_i}$ and the raw intermediate frame $\mathbf{x}_{k, i}^{\bar{t}_i}$, $\lVert \nabla \mathbf{\bar{x}}_{k, i}^{\bar{t}_i}-\nabla \mathbf{x}_{k, i}^{\bar{t}_i} \rVert _1$ is the L1-norm of the gradient difference between the predicted intermediate frame $\mathbf{\bar{x}}_{k, i}^{\bar{t}_i}$ and the gradient of the raw intermediate frames $\mathbf{x}_{k, i}^{\bar{t}_i}$, and $\kappa$ is a weight parameter to balance the importance of these two differences on TAIN model training. 

\section{Problem Formulation}
\label{sec: problem}
In this section, we first introduce a new semantic metric to capture the semantic communication performance in the SFMA system. Then, we introduce the considered optimization problem to jointly maximize the sum semantic transmission rate and minimize the temporal gap between the simultaneously transmitted frames. 

\subsection{Semantic Communication Performance Metric}
\label{sec: interference factor}
To optimize the data rate of the proposed SFMA system, it is necessary to derive the appropriate data rate equation for our considered SFMA since standard channel capacity in \cite{ref1} fails to capture the complex interactions and interference between signals at the semantic level. For example, suppose the channel noise power is $\sigma^2$, when we use standard communication techniques to transmit frame $\boldsymbol{x}^{t_1}_{k,1}$ and $\boldsymbol{x}^{t_2}_{k,2}$ to users 1 and 2, respectively, the signal-to-interference-plus-noise ratio (SINR) of the received signal $\boldsymbol{y}_{k,1}$ is expressed as
\begin{equation}
    \gamma_{k,1}\left( p_{k,1},p_{k,2} \right)=\frac{p_{k,1}\left| h_{k,1} \right|^2}{p_{k,2}\left| h_{k,1} \right|^2+\sigma^{2}}.
    \label{eq: conventional SINR}
\end{equation}
Therefore, we cannot use (\ref{eq: conventional SINR}) to represent the SINR of our considered semantic communications. Next, we further use a simulation to explain why (\ref{eq: conventional SINR}) cannot be used for semantic communications. In particular, to show the actual SINR of semantic communications, we can use the mean square error (MSE) between the original frame $\boldsymbol{x}^{t_1}_{k,1}$ and the recovered frame $\boldsymbol{\hat{x}}^{t_1}_{k,1}$ to represent the SINR of the signal $\boldsymbol{x}^{t_1}_{k,1}$. Hence, we have
\begin{equation}
    \gamma_{k,1}\left( p_{k,1},p_{k,2} \right)=\frac{p_{k,1}\left| h_{k,1} \right|^2}{d (\boldsymbol{x}^{t_1}_{k,1}, \boldsymbol{\hat{x}}^{t_1}_{k,1})},
    \label{eq: SINR2}
\end{equation}
where $d (\boldsymbol{x}^{t_1}_{k,1}, \boldsymbol{\hat{x}}^{t_1}_{k,1})$ is the MSE between the original frame $\boldsymbol{x}^{t_1}_{k,1}$ and the recovered frame $\boldsymbol{\hat{x}}^{t_1}_{k,1}$.  Fig. \ref{fig: compare SINR} shows how SINR changes as SNR varies when $p_{k,1}=p_{k,2}$. This figure shows that a gap exists between the SINR resulting from standard communications utilizing the MMSE algorithm and semantic communications in the framework of combined JSCC since the SINR equation used in standard communications does not consider semantic features. 
\begin{figure}[t]
    \centerline{\includegraphics[width=0.45\textwidth]{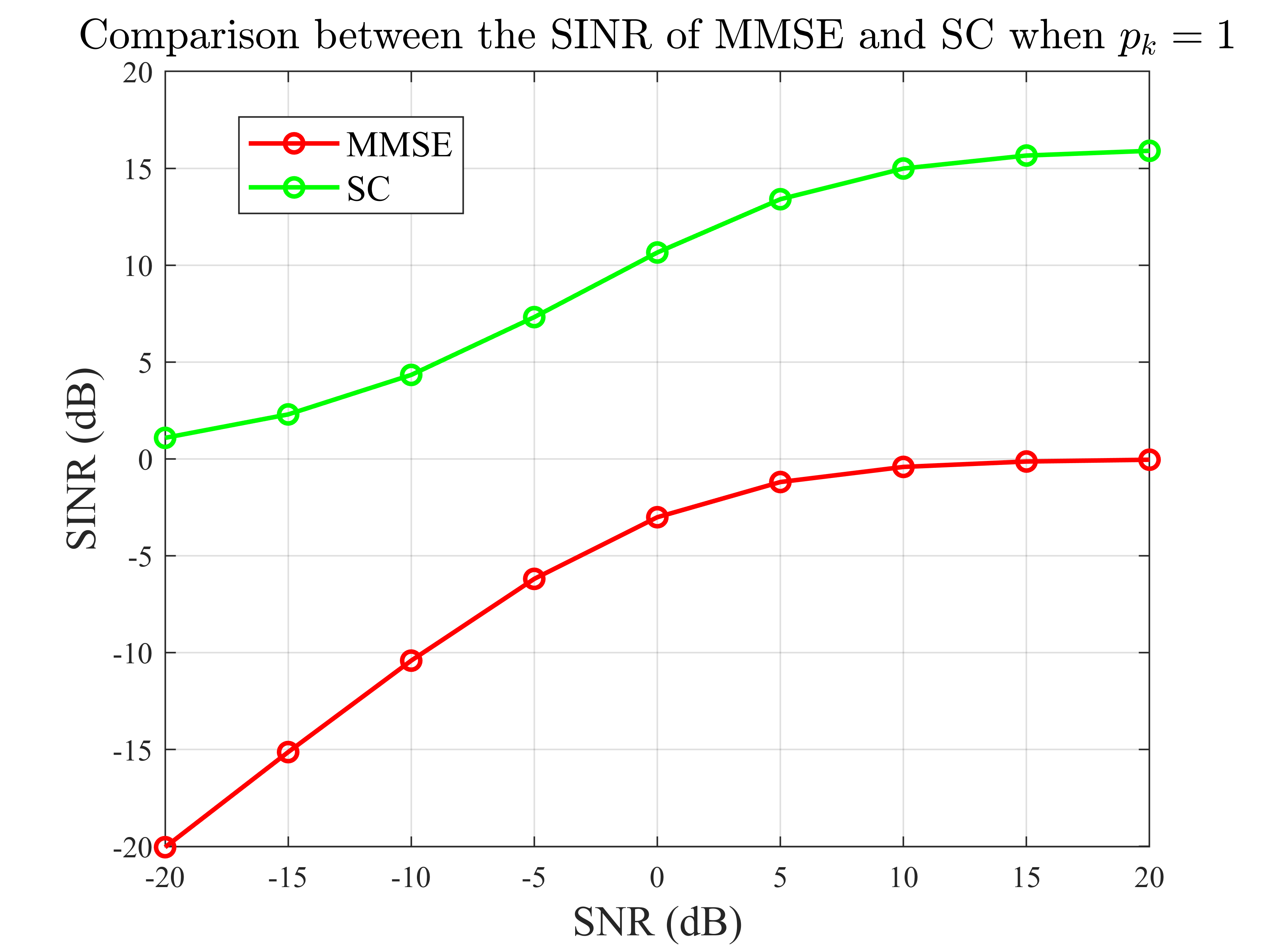}}
    \captionsetup{font={small}}
    \caption{The comparison of SINR achieved by MMSE and the proposed semantic communication system.}
    \label{fig: compare SINR}
\end{figure}
\begin{figure}[t]
    \centerline{\includegraphics[width=0.45\textwidth]{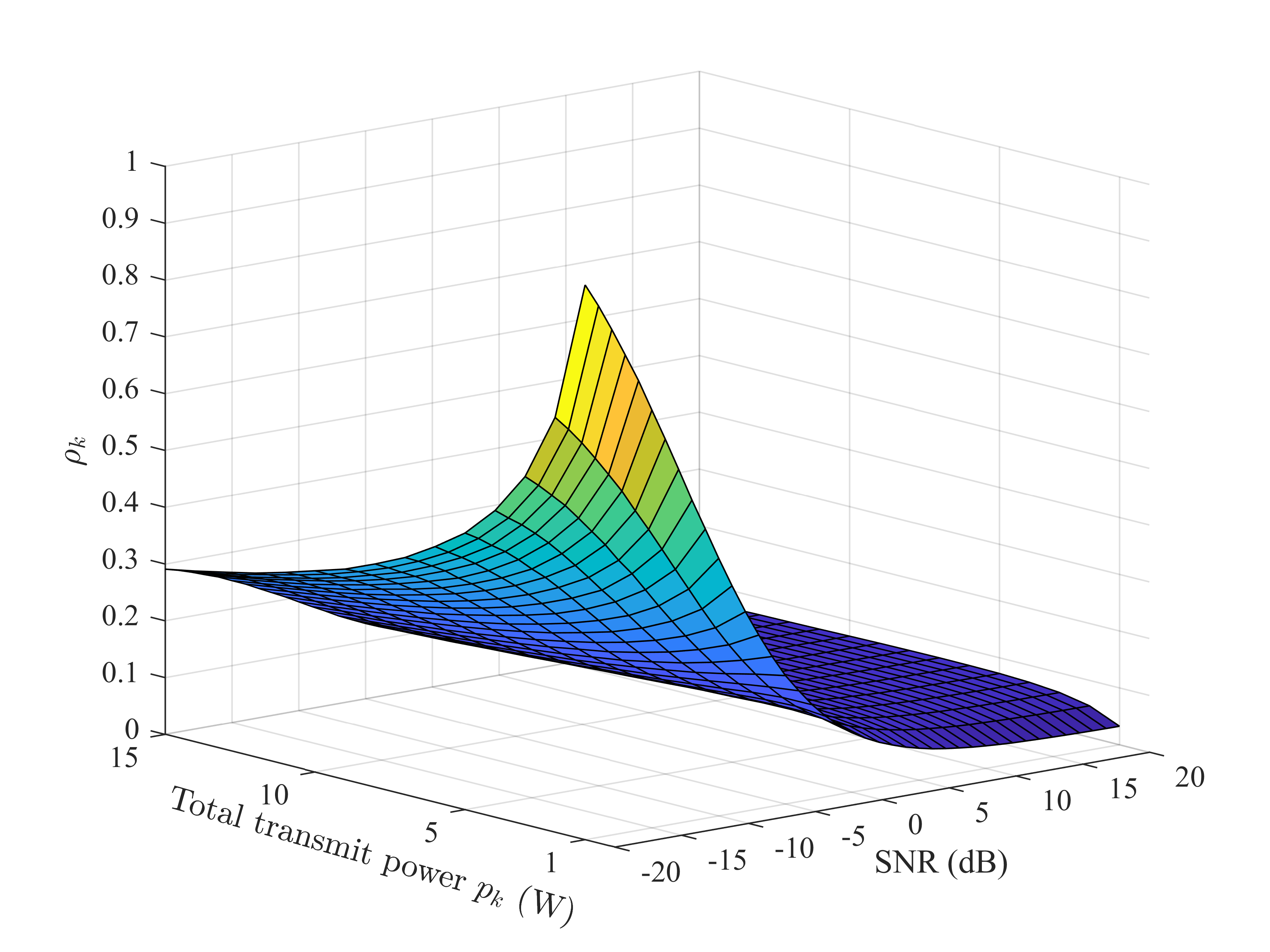}}
    \captionsetup{font={small}}
    \caption{The impact of the semantic interference factor $\rho_k$ on SNR for two users in the proposed semantic communication system.}
    \label{fig: rho 3D}
\end{figure}

To accurately represent the SINR of semantic communications, we introduce a semantic interference factor $\rho$ into \eqref{eq: conventional SINR} such that the SINR of the signal $\boldsymbol{x}_{k,1}$ is denoted by
\begin{equation}
    \gamma_{k,1}\left( p_{k,1},p_{k,2}\right)=\frac{p_{k,1}\left| h_{k,1} \right|^2}{\rho _{21}^{k}\left(p_{k,1},p_{k,2}\right) p_{k,2}\left| h_{k,1} \right|^2+\sigma^{2}},
\label{eq: SINR with rho}
\end{equation}
where $\rho^k _{21}$ captures the impact of signal $\boldsymbol{x}_{k,2}$ on signal $\boldsymbol{x}_{k,1}$ within the semantic space. In SFMA systems, the BS usually allocates a fixed amount of signal transmission power to each user group according to an optimal power allocation policy. Therefore, \eqref{eq: SINR with rho} demonstrates that $\gamma_{k,1}\left( p_{k,1},p_{k,2}\right)$ is a function of the total allocated power in group $k$, i.e., $p_{k,1} + p_{k,2}$. This dependence highlights that the interference between the users in the same group is not solely a function of both physical layer parameters and the semantic interactions between their transmissions. Therefore, to maximize the overall data rate of the SFMA system, it is essential to first consider a group-level power allocation strategy across multiple groups and then optimize the power allocation within each group.

By strategically allocating total power, i.e., $p_{k,1}+p_{k,2}$ among the groups and subsequently optimizing the distribution between $p_{k,1}$ and $p_{k,2}$ within each group, the system can effectively enhance the SINR for each group users and minimize semantic interference between them. This layered optimization, starting at the group level and followed by power allocation within each group, ensures that the total system data rate is maximized while taking into account both physical and semantic interference dynamics.
Fig. \ref{fig: rho 3D} shows how $\rho_{21}^{k}$ varies with $\left(p_{k}, \text{SNR}_{k,1}\right)$ when $p_{k,1}=p_{k,2}$. This relationship can be represented as a function of $p_{k, 1}$ and $p_{k,2}$, i.e., $\rho^k_{21}=\rho^k_{21}\left(p_{k, 1}, p_{k,2}\right)$.
Given (\ref{eq: SINR with rho}), the achievable rate of user $i$ in group $k$ is 
\begin{align}
    r_{k,i}\left( p_{k,1},p_{k,2} \right)=\log _2\left( 1+ \gamma_{k,i}\left( p_{k,1},p_{k,2}\right)
    \right). 
\label{eq: r_k}
\end{align}
Therefore, the sum rate in user group \textit{k} is calculated as
\begin{equation}
    R_{k,12}\left( p_{k,1},p_{k,2} \right) = r_{k,1}\left( p_{k,1},p_{k,2} \right) + r_{k,2}\left( p_{k,1},p_{k,2} \right).
\end{equation}

\subsection{Problem Formulation}

To optimize the user pairing strategy and the transmit power $p_{k, i}$ allocated to each user, we formulate an optimization problem to maximize the total sum rate of the system and minimize the temporal gap between the simultaneously transmitted frames in one group while satisfying the minimum data rate requirement of each user. Let $a_{k,m}\in \left\{ 0,1 \right\}$ be an index that indicates whether user $m$ is assigned to group $k$.
The optimization problem can be developed as
\begin{subequations}
\begin{align}
    \underset{\boldsymbol{a}, \boldsymbol{P}}{\max} & \sum_{k=1}^K{\sum_{i=1}^M{\sum_{j=1,j\ne i}^M{a_{k,i}a_{k,j}\left( R_{k,ij}\left( p_{k,1},p_{k,2} \right)-\alpha D_{k,ij} \right)}}}, \label{eq: Problem} \tag{12} \\ 
    \text{s.t.} \; \; & \sum_{k=1}^K{a_{k,m}=1, \; \forall m\in \mathcal{M},} \label{eq: akm 1} \\
    & \sum_{m=1}^M{a_{k,m}=2}, \forall k\in \mathcal{K}, \label{eq: akm 2} \\
    & a_{k,i}a_{k,j}D_{k,ij}\le \Delta , \; \forall k\in \mathcal{K}, \; \forall i,j\in \mathcal{M},i\ne j, \label{eq: a 1} \\
    & a_{k,i},a_{k,j}\in \left\{ 0,1 \right\} , \; \forall k\in \mathcal{K}, \; \forall i,j\in \mathcal{M}, \label{eq: a 2} \\
    & r_{k,i}\left(p_{k,1},  p_{k,2}\right) \ge R_{k,i}, \; \forall k\in \mathcal{K}, \; i\in\{1,2\}, \label{eq: R constraints} \\
    & \sum_{k=1}^K \left(p_{k,1}+p_{k,2}\right) \le P_{\max}, \label{eq: power constraints 1} \\
    & {p}_{k,i} \ge 0, \; \forall k\in \mathcal{K}, \; i\in\{1,2\}, \label{eq: power constraints 2}
\end{align}
\end{subequations}
where $D_{k,ij}$ is the temporal gap between user $i$ and user $j$ in group \textit{k},
$\alpha$ is the weight parameter to balance the impact of the achievable transmission rate and the temporal gap, $\Delta$ denotes the maximum tolerable temporal gap which ensures that the GAI-based model employed at the users for frame interpolation can generate an acceptable intermediate frame, $R_{k, i}$ is the minimum transmit rate requirement of user $i$ in group $k$, $\boldsymbol{P}=\left(\boldsymbol{p}_1, \cdots, \boldsymbol{p}_K \right) ^T$ is the power allocation matrix, $\boldsymbol{p}_k=\left( p_{k,1},p_{k,2} \right)^T, \forall k \in \mathcal{K}$ is the power allocation vector, $\boldsymbol{a}=\left( \left(a_{1,1}, a_{1,2}\right), \cdots, \left(a_{K,1}, a_{K,2}\right) \right)$ is the user pairing matrix, and $P_{\max}$ is the maximum transmit power of the BS. 

Solving problem \eqref{eq: Problem} faces several challenges. First, the power $p_{k,1}$ and $p_{k,2}$ used for semantic information transmission in each group $k$ are coupled in both the objective function and the constraints (i.e., \eqref{eq: R constraints} and \eqref{eq: power constraints 1}). Second, the problem is non-convex due to the interaction between the power allocation and the semantic interference factors. Third, the problem involves optimizing the power allocation for multiple groups simultaneously. As the number of groups increases, the computational complexity of solving the problem increases. Therefore, to solve problem \eqref{eq: Problem}, we propose a three-step iterative algorithm that decomposes the problem in \eqref{eq: Problem} into three interrelated sub-problems and iteratively optimizes these problems to find the global solution. 
This problem decomposition method can significantly reduce the complexity of solving the overall optimization problem.

\section{Solution for the Optimization Problem}
\label{sec: solutions}
To solve the problem in \eqref{eq: Problem}, we decompose it into three sub-problems: 1) user pairing strategy, 2) intra-group power allocation, and 3) inter-group power allocation. Next, we first introduce a user pairing algorithm to maximize the sum transmission rate of the SFMA system while minimizing the temporal gap between simultaneously transmitted frames. Second, we optimize the inter-group power allocation in the SFMA system aiming to maximize the sum data rate of all users while satisfying each user's minimum data rate requirement. Third, we optimize the intra-group power allocation based on the optimized inter-group power allocation strategy to further improve each user's semantic transmission rate within each group. 
Our proposed method that solves the problem \eqref{eq: Problem} using three steps can simplify the optimization process and effectively reduce the complexity. 

\subsection{User Pairing Algorithm Design} \label{sec: user pairing}

Given the fixed power allocation for intra-group and inter-group users, the optimization problem in \eqref{eq: Problem} can be rewritten as

\begin{subequations}
\begin{align}
    \underset{\boldsymbol{a}}{\max} & \sum_{k=1}^K{\sum_{i=1}^M{\sum_{j=1,j\ne i}^M{a_{k,i}a_{k,j}\left( R_{k,ij}\left( p_{k,1},p_{k,2} \right)-\alpha D_{k,ij} \right)}}}, \label{eq: user pairing} \tag{13} \\ 
    \text{s.t.} \; \; & \eqref{eq: akm 1} - \eqref{eq: a 2}. \notag
\end{align}
\end{subequations}

To solve this problem, we propose to use the Gale-Shapley algorithm in \cite{ref47}. Since the Gale-Shapley algorithm in \cite{ref47} is a matching algorithm, 
we define user pairing as a matching process. In particular, if users $i$ and $j$ are paired as a group, we state that they are matched with each other. Then, we need to define the preference value and the matching condition that will be used in the Gale-Shapley algorithm. The preference value $V_{k,ij}$ between users $i$ and $j$ to form group $k$ is denoted by
\begin{equation}
    V_{k, ij}\left( p_{k,i},p_{k,j} \right) = R_{k,ij}\left( p_{k,i},p_{k,j} \right)-\alpha D_{k,ij}.
\end{equation}
Based on the perfect channel state information, the preference list $L(i)$ of user $i$ paired in group $k$ is denoted by
\begin{equation}
    L\left( i \right) = \left[ V_{k,i1}\left( p_{k,i},p_{k,1} \right), \cdots, V_{k,iM}\left( p_{k,i},p_{k,M} \right)  \right].
\end{equation}
The preference list $L\left( i \right)$ contains the preference values between user $i$ and other users in the system. We say user $l$ prefers user $i$ to user $j$ if 
$V_{il}>V_{jl}$.
Then, we illustrate the matching condition. Suppose that user $j$ has been paired with user $l$ and user $i$ is requesting to be paired with user $j$, user $j$ will choose to pair with user $i$ and dispose of user $l$, if and only if the following matching condition is satisfied:
\begin{equation}
    V_{ij}>V_{jl} \land D_{ij}\le \Delta.
\end{equation}
This pairing process continues iteratively until all users have been paired. 
The output of the algorithm is the list of user pairs $\mathcal{S}$.
This approach ensures that each user forms an optimal pairing based on their preferences so that the sum semantic transmission rate is maximized while meeting the temporal gap threshold.

\subsection{Inter-group Power Allocation Optimization in the SFMA System}
\label{sec: power allocation among groups}
After the user pairing scheme is determined, we design a power allocation algorithm for the SFMA system. To optimize the power allocated to each group,  we first let $p_k=p_{k, 1}+p_{k,2}$ to represent the power of group $k$. Fix the power allocation factor, the semantic interference factor for each user group is expressed as a function of $p_k$, i.e., $\rho^k_{ji}=\rho^k_{ji}\left( p_k \right)$. By substituting \eqref{eq: r_k} into \eqref{eq: R constraints}, the optimization problem in \eqref{eq: Problem} can be reformulated as
\begin{subequations} 
\label{eq: overall problem2}
    \begin{align}
    \underset{\boldsymbol{p}}{\max} & \;  \sum_{k=1}^{K} R_{k,12}\left( p_{k,1},p_{k,2} \right), \tag{17} \label{eq: subproblem 1}  \\
    \text{s.t.} \; \; & \left(\rho^k_{ji} \left(p_{k} \right) p_{k,j} + p_{k,i}\right)\left|h_{k,i}\right| + \sigma^2_{k,i} \ge \notag\\
    & 2^{R_{k,i}} \left( \rho^k_{ji} \left(p_{k} \right) p_{k,i}\left|h_{k,i}\right| + \sigma^2_{k,i} \right),\notag\\
    & \forall k\in \mathcal{K}, i \ne j, \; i,j\in\{1,2\}, \label{eq: R constraints extent} \\
    & \sum_{k=1}^K p_{k} \le P_{\max}, \label{eq: BS power} \\
    &  \eqref{eq: power constraints 2}, \; \boldsymbol{p} \ge 0, \label{eq: power constraints}
    \end{align} 
\end{subequations}
where $\boldsymbol{p} = \left(p_1, \cdots, p_K \right)^T$ is the inter-group power allocation vector. 
\begin{table*}[!t]
\centering
\begin{minipage}{1\textwidth}
    \begin{align}\label{long equation 1}
  &L\left( p_k,\boldsymbol{\lambda }_1,\boldsymbol{\lambda }_2,\mu \right) =\sum_{k=1}^{K}{\left[ r_{k,1}\left( p_{k,1}, p_{k,2} \right)+r_{k,2}\left(p_{k,1}, p_{k,2} \right) \right] + \mu \left( P_\text{max}-\sum_{k=1}^{K}{p_k} \right)} \notag
  \\
  & \quad \quad \quad \quad \quad \quad \quad +\sum_{k=1}^{K} {\lambda _{k,1}\left[ \left( p_{k,1}+\rho _{21}^k\left( p_{k} \right)p_{k,2} \right) \left| h_{k,1} \right|^2+\sigma _{k,1}^{2}-2^{R_{k,1}}\left( \rho _{21}^k\left( p_{k}\right)p_{k,2}\left| h_{k,1} \right|^2+\sigma _{k,1}^{2} \right) \right]}
 \notag
  \\
  & \quad \quad \quad \quad \quad \quad \quad +\sum_{k=1}^{K} {\lambda _{k,2}\left[ \left( p_{k,2}+\rho _{12}^k\left( p_{k} \right)p_{k,1} \right)\left| h_{k,2} \right|^2+\sigma _{k,2}^{2}-2^{R_{k,2}}\left( \rho _{12}^k\left( p_{k}\right)p_{k,1}\left| h_{k,2} \right|^2+\sigma _{k,2}^{2} \right) \right]}.
\end{align}
\medskip
\hrule
\end{minipage}
\end{table*}

We then introduce the Lagrangian penalty of the problem (\ref{eq: subproblem 1}), as shown in \eqref{long equation 1}, where $\boldsymbol{\lambda }_{1}=\left( \lambda _{1,1}, \cdots , \lambda _{K,1} \right) ^T, \boldsymbol{\lambda }_{2}=\left( \lambda _{1,2}, \cdots , \lambda _{K,2} \right) ^T, \mu$ are non-negative Lagrange multipliers associated with the constraints \eqref{eq: R constraints extent} and \eqref{eq: BS power}. Then, we can find the extreme points of the allocated power $p_k$ in each group $k$, as shown in the following Lemma. 

\begin{lemma}
    For each group $k$, the extreme points of the allocated power $p_k$ in problem \eqref{eq: overall problem2} are
\begin{align} \label{eq: pk1}
    p_{k_1}=\frac{\sigma _{k,1}^{2}\left( 2^{R_{k,1}}-1 \right)}{\left| h_{k,1} \right|^2\left( \eta _{k,1}+\eta _{k,2}\rho _{21}^k\left(p_{k_1} \right)\left( 1-2^{R_{k,1}} \right) \right)},
\end{align}
\begin{align} \label{eq: pk2}
    p_{k_2}=\frac{\sigma _{k,2}^{2}\left( 2^{R_{k,2}}-1 \right)}{\left| h_{k,2} \right|^2\left( \eta _{k,2}+\eta _{k,1}\rho _{12}^k\left(p_{k_2} \right)\left( 1-2^{R_{k,2}} \right) \right)},
\end{align}
 and $p_{k3}$ must satisfy the following equation:
\begin{align}\label{eq: pk3}
    \frac{\partial \left[ r_{k,1}\left( p_{k,1}, p_{k,2} \right) +r_{k,2}\left(p_{k,1}, p_{k,2} \right) \right]}{\ln 2 \cdot \partial p_k} -\mu = 0.
\end{align} 
\end{lemma}

\begin{proof}
    The proof of Lemma 1 is at \appendixautorefname.
\end{proof}

Using Lemma 1, we can determine the optimal inter-group power allocation vector $\boldsymbol{p}^*$ through Algorithm \ref{Algorithm power allocation 1}.
First, we iteratively update the value of $\mu$ and compute the value of three extreme points $p_{k_1}, p_{k_2}, p_{k_3}$ to find the optimal value $\mu_\text{opt}$, which allows the sum power of the users to satisfy the BS power constraint \eqref{eq: BS power}. Second, we substitute the value of $\mu_\text{opt}$ into Lemma 1 and compute the maximum allocated power among $p_{k_1}, p_{k_2}, p_{k_3}$ for each user group to obtain the optimal inter-group power allocation vector $\boldsymbol{p}^*$.

\begin{algorithm}[t]
    \caption{The optimal inter-group power allocation}
        \textbf{Input: }Channel coefficient $h_{k, i}$, minimum transmit requirement $R_{k, i}$, noise power $\sigma_{k,i}^2$ in each user group $k$, number of groups $K$ and the maximum power of BS $P_\text{max}$.
        
        \textbf{Output: }The optimal Lagrange multiplier $\mu_\text{opt}$ and the inter-group power allocation vector $\boldsymbol{p}^*$.
        
        \textbf{Function: } Iterate to find the optimal $\mu_\text{opt}$ and $\boldsymbol{p}^*$
        
        \begin{algorithmic}[1]
            \STATE \textbf{Initialization: }Set the initial guess for $\mu$ and the allowable tolerance $\epsilon$.
            \WHILE{$|\sum_{i=1}^{k} p_i - P_\text{max}| \ge \epsilon$}
                \FOR{each $k = 1:K$}
                    \STATE set $p_{k,1}=p_{k,2}$.
                    \STATE compute $p_{k_1}, p_{k_2}, p_{k_3}$ through \eqref{eq: pk1}, \eqref{eq: pk2} and \eqref{eq: pk3}.
                    \IF{$p_{k_1}, p_{k_2}, p_{k_3}$ are feasible}
                        \STATE set $p_k=\text{max}\{p_{k_1}, p_{k_2}, p_{k_3}\}$.
                    \ENDIF
                    \STATE $\sum_{i=1}^{k} p_i = \sum_{i=1}^{k-1} p_i + p_k$.
                \ENDFOR
                \STATE Update $\mu$.
            \ENDWHILE
        \end{algorithmic}
\label{Algorithm power allocation 1}
\end{algorithm}

\subsection{Intra-group Power Allocation Optimization in the SFMA System}
\label{sec: power in one group}
Given the optimized inter-group power allocation vector $\boldsymbol{p}^*$,
the optimization problem in \eqref{eq: subproblem 1} is simplified as
\begin{subequations}\label{eq: optimal eta}
  \begin{align}
    \underset{\boldsymbol{p}_1, \cdots, \boldsymbol{p}_K}{\max} \; & \sum_{k=1}^{K}   R_{k,12}\left( p_{k,1},p_{k,2} \right), \tag{22} \\
    \text{s.t.} \; & \sum_{i=1}^2{p_{k,i}}=p_k, \; \forall k \in \mathcal{K}. \label{eq: eta limit}
  \end{align}
\end{subequations}
From \eqref{eq: optimal eta}, we see that each $\boldsymbol{p}_k$ is only dependent on $r_{k,1}\left( p_{k,1}, p_{k,2}\right) + r_{k,2}\left(p_{k,1}, p_{k,2} \right)$. Hence, we can individually optimize the power $p_{k,1}$ and $p_{k,2}$ within each group, and the problem in \eqref{eq: optimal eta} can be divided into $K$ sub-optimization problems. Meanwhile, from the constraint 
\eqref{eq: eta limit}, we have $p_k=p_{k,1}+p_{k,2}$. Hence, we can use $p_k-p_{k,1}$ to represent $p_{k,2}$. Therefore, the problem in \eqref{eq: optimal eta} can be simplified as  
\begin{equation}\label{eq: J}
\underset{{p}_{k,1}}{\max} \left( r_{k,1}\left( p_{k,1}, p^*_k-p_{k,1}\right) + r_{k,2}\left(p_{k,1}, p^*_k-p_{k,1} \right) \right).
\end{equation}
Since the objective function in \eqref{eq: J} is a sum of two concave logarithmic functions with respect to \(p_{k,1}\). 
Thus, the overall objective function in (\ref{eq: J}) is concave. In consequence, standard optimization methods such as gradient descent \cite{ref50} can be used to find the optimal \(p^*_{k,1}\). Given $p^*_{k,1}$, we can directly obtain $p^*_{k,2}=p^*_k-p^*_{k,1}$. Algorithm \ref{Algorithm power allocation 2} summarizes the entire procedure to solve the problem in \eqref{eq: subproblem 1}. 

\begin{algorithm}[t]
    \caption{Proposed overall power allocation in the SFMA system}
        \textbf{Input: }Channel coefficient $h_{k, i}$, minimum transmit requirement $R_{k, i}$, noise power $\sigma_{k,i}^2$ in each user group $k$, number of groups $K$ and the maximum power of BS $P_\text{max}$.
        
        \textbf{Output: }Optimal inter-group power allocation vector $\boldsymbol{p}^*$, and optimal intra-group power allocation vector $\boldsymbol{p}_k$ for each group $k$.
        
        \begin{algorithmic}[1]
            \STATE Obtain the optimal inter-group power allocation vector $\boldsymbol{p}^*$ using Algorithm \ref{Algorithm power allocation 1}.
            \STATE Given $\boldsymbol{p}^*$, solve the optimization problem \eqref{eq: J} using the gradient descent method to find the optimal intra-group power allocation vector $\boldsymbol{p}_k$ for each group $k$.
        \end{algorithmic}
\label{Algorithm power allocation 2}
\end{algorithm}

\subsection{Convergence and Complexity Analysis}

Based on the above analysis, the overall algorithm is summarized in Algorithm \ref{Algorithm overall}. 
The convergence of the proposed algorithm is guaranteed at each stage. In the user pairing stage, the user pairing algorithm developed from the Gale–Shapley algorithm ensures a stable matching in a finite number of steps, where no two users would prefer each other over their current match, ensuring termination. In the inter-group power allocation stage, convergence to a feasible power allocation is achieved by iteratively solving the Karush–Kuhn–Tucker (KKT) conditions and determining the Lagrange multiplier $\mu$. The extreme points for power allocation are evaluated, and the water-filling algorithm ensures the total power constraint \eqref{eq: BS power} is satisfied, leading to convergence. Finally, in the intra-group power allocation stage, the optimal intra-group power allocation strategy is obtained by one-dimensional searching within finite iterations, which ensures convergence to the optimal solutions. Overall, the proposed algorithm is designed to guarantee convergence in each stage, making it robust and effective for user pairing and power allocation in SFMA systems.

The complexity of the proposed algorithm consists of three primary stages: user pairing, inter-group power allocation, and intra-group power allocation. The user pairing stage computes the preference matrix for all user pairs and generates preference lists based on the calculated preference values, resulting in a complexity of $O(M^2log(M))$ due to the sorting step. The user pairing algorithm is then used for pairing, with a worst-case complexity of $O(M^2)$, making the total complexity for the user pairing stage $O(M^2logM)$. In the inter-group power allocation stage, the algorithm solves a convex optimization problem by deriving extreme points from the KKT conditions and then selecting the maximum power allocation for each group. The complexity of calculating the Lagrange multiplier $\mu$ is $O(K)$ per iteration, and determining the three potential power allocations $p_{k_1}, p_{k_2}, p_{k_3}$ for each group also requires $O(K)$. The water-filling algorithm then selects the optimal power allocation while ensuring the total power constraint \eqref{eq: BS power}, resulting in overall complexity for this stage of $O(K) \approx O(M)$, given $M=2K$. In the intra-group power allocation stage, the algorithm optimizes the intra-group power allocation within each group and has a complexity of $O(TK)$, where $T$ is the number of iterations, leading to $O(TM)$ for $M=2K$. Combining these complexities, the total computational complexity of the proposed algorithm is $O(M^2logM+M+TM)$. In addition, when $M$ is extremely larger, the $M^2logM$ term dominates the total complexity and the total computational complexity is $M^2log(M)$.

\begin{algorithm}[t]
	\caption{User Pairing and Power Allocation in SFMA System}
        \textbf{Input}: $\mathcal{M}$, $R_\text{min}$, $\Delta$, $h_{k,m}$, $P_\text{max}$ and weight $\alpha$.
		
        \textbf{Output}: User pairing set $\mathcal{S}$ and the optimal power allocation matrix $\boldsymbol{P}^*$.
        
        \begin{algorithmic}[1]
    		\STATE Fix $\boldsymbol{P}$, and then implement Algorithm \ref{Algorithm power allocation 1} to form the pairing set $\mathcal{S}$.
    	    \STATE Apply $\mathcal{S}$ in problem \eqref{eq: subproblem 1}, and then implement Algorithm \ref{Algorithm power allocation 2} to find the optimal $\boldsymbol{P}^*$.
        \end{algorithmic}

\label{Algorithm overall}
\end{algorithm}

\section{Simulation Results and Analysis}
\label{sec: results}
For our simulations, we consider users to be uniformly distributed in the $500$ m $\times$ $500$ m area with a BS located at the center. The propagation model is $L(d)=37+30\text{log}(d)$, with the standard deviation of shadow fading $4$dB. AWGN noise power is $\sigma^2=-104\text{dBW}$. All users have a minimum rate requirement of $1$. The CIFAR-10 dataset is used to train both semantic encoders and decoders. The training dataset comprises $50,000$ $32\times32$ images from CIFAR-10, combined with random realizations of the channel under consideration. For testing, we use 10,000 images from the CIFAR-10 testing dataset. The model is trained with Adam optimizer over $300$ epochs, starting with a learning rate of $1\times10^{-3}$, reduced $1/10$ every $100$ epochs. The batch size is $128$, and the compression ratio for two frames is both $\gamma=1/3$. We also trained our communication model at various SNR levels using the SNU Frame Interpolation with Large Motion (SNU-FILM) dataset \cite{ref41}. A pre-trained TAIN model is leveraged to process the reconstructed frames and generate intermediate frames. The SNU-FILM dataset encompasses videos with diverse motion magnitudes, stratified into four settings: \textit{Easy}, \textit{Medium}, \textit{Hard}, and \textit{Extreme}, based on the temporal gap between frames. Each triplet in the dataset includes a start frame, an end frame, and a target frame. For example, in the \textit{Medium} dataset mode,  every fourth frame remains and the middle frame among the three dropped ones is employed as the target, corresponding to the scheme where the temporal gap is $\Delta=4$. 

To compare the performance of different MA schemes, we use three baselines: a) a fixed power allocation in the NOMA system called F-NOMA, b) an O-JSCC that each user in one group shares half of the bandwidth with equal power allocation to transmit the signal through JSCC networks separately, and c) an OFDMA scheme where the BS allocates the equal subchannel to different users in the system. To simplify the problem, all baselines adopt the same user pairing strategies introduced in \cite{ref26} where the BS will pair two users whose channel conditions are more distinctive.
We then compare transmission performance between our proposed JSCC network with a combined JSCC network with one joint source channel (JSC) encoder and two JSC decoders in one user group to show the effectiveness of our networks in recovering the images.
To compare the impact of the temporal gap, we also evaluate the model with the dataset in \textit{Medium} and \textit{Extreme} modes, corresponding to the schemes where the temporal gap is $\Delta=4$ and $\Delta=16$, respectively. Both frames were transmitted with a compression ratio $\gamma=1/3$. 

We have also further extended our work to the three-user case and run experiments to evaluate the performance of three different users. In the simulation, we assume that three users require different CIFAR-10 images simultaneously and the information they require is superimposed and sent to each user through different links.

\subsection{User Pairing and Power Allocation Performance}

\begin{figure}[t]
    \centerline{\includegraphics[width=0.45\textwidth]{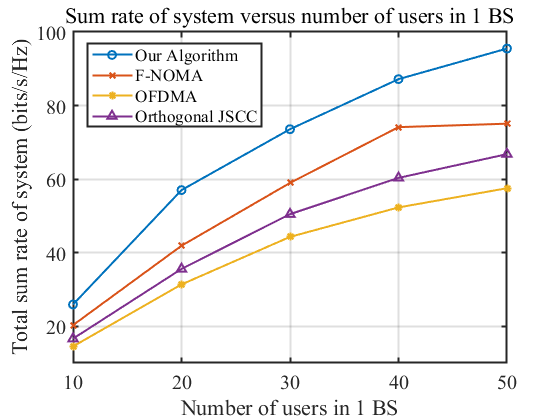}}
    \captionsetup{font={small}}
    \caption{Total sum rate of the system versus different number of users.}
    \label{fig: users vs rate}
\end{figure}

\begin{figure}[t]
    \centerline{\includegraphics[width=0.45\textwidth]{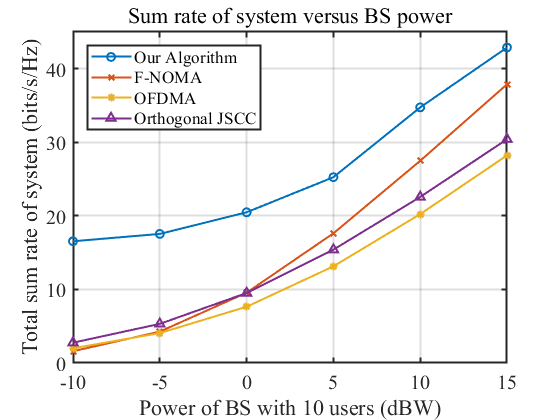}}
    \captionsetup{font={small}}
    \caption{Total sum rate of the system versus different power of BS with 10 users}
    \label{fig: power vs rate}
\end{figure}

In Fig. \ref{fig: users vs rate}, the performance of the total sum rate is evaluated with the number of users $M \in \left[10, 60\right]$. This figure demonstrates that the total sum rate increases as the number of users grows. However, as the number of users in the BS becomes larger, the growth rate of the sum rate slows down due to the limited power of the BS. From observing Fig. \ref{fig: users vs rate}, we conclude that our proposed user pairing and power allocation algorithm in the SFMA system outperforms the F-NOMA scheme. For instance, when the number of users is $30$, the total sum rate of our proposed algorithm is $24.8\%$ higher than that of the F-NOMA scheme, and $66.1\%$ higher than the OFDMA. This improvement is due to the fact that the F-NOMA scheme implements a fixed power allocation factor in each user group and does not consider semantic interference when allocating the power. 
Our algorithm also performs $45.8\%$ better than the O-JSCC scheme when there are $30$ users in the system. This is due to the fact that although the JSCC network will enhance the data rate in the O-JSCC scheme, each user in a group is equally allocated half of the bandwidth to avoid interference, resulting in a decreased total sum rate, as expected from Shannon's formula in calculating the sum rate. Our algorithm also achieves $66.1\%$ higher data rate than the OFDMA when the number of users is $30$.

In Fig. \ref{fig: power vs rate}, we investigate the sum rate versus the power of the BS with 10 users in the system. Our Algorithm exhibits a consistently high total sum rate across all power levels. This is due to the fact that we consider semantic interference when designing the user pairing and power allocation strategies. The F-NOMA scheme shows a lower total sum rate at lower power levels but demonstrates rapid improvement as the BS power increases. This is because the power allocation strategy in F-NOMA allows weak channel users to utilize channel resources at higher power levels better, leading to significant sum rate gains. The OFDMA scheme performs the lowest total sum rate across all power levels, and the O-JSCC performs slightly better. This is due to the fact that the users in the OFDMA scheme can only be assigned to one subchannel, which limits the achievable data rate. We should note that owing to the relatively low computational and hardware requirements, OFDMA remains widely employed in contemporary wireless networks, balancing system performance with the practicality of hardware implementation.

\subsection{Transmission Performance}
 This evaluation aims to demonstrate the capability of the network to transmit the superimposition of two images simultaneously. We compare this approach with the scheme that utilizes one JSC encoder and two JSC decoders for transmitting the superimposition of two images. The peak signal-to-noise ratio (PSNR) is used to evaluate the ratio of the maximum signal power to the noise power that distorts the signal, which is expressed as
\begin{equation}  
    \text{PSNR}= 10 \log_{10}\frac{\text{MAX}^2}{\text{MSE}}\left(\text{dB}\right),   
\end{equation}
where $\text{MSE}=d\left(\boldsymbol{x}, \boldsymbol{\Tilde{x}}\right)$ is the mean squared error between the raw images $\boldsymbol{x}$ and the reconstructed images $\boldsymbol{\Tilde{x}}$, and $\text{MAX}=255$ is the maximum possible value of the image pixels. 
Since the CIFAR-10 dataset consists of 24-bit depth RGB images, and there are 3 channels: R, G, and B, we have $\text{MAX}=2^8-1=255$.
Fig. \ref{fig: PSNR 1} demonstrates that the Attention-based JSCC scheme performs effectively in an AWGN channel, achieving high PSNR values for both users across a range of test SNRs. This is due to the fact that the attention mechanism in our network architecture strengthens the dependency between different channels at the pixel level by weighting them. Meanwhile, we notice that the performance of user 1 and user 2 is almost the same. This indicates the fairness and effectiveness of the proposed scheme in maintaining similar quality for multiple users, thus demonstrating the robustness of the scheme.

\begin{figure}[t]
    \centerline{\includegraphics[width=0.45\textwidth]{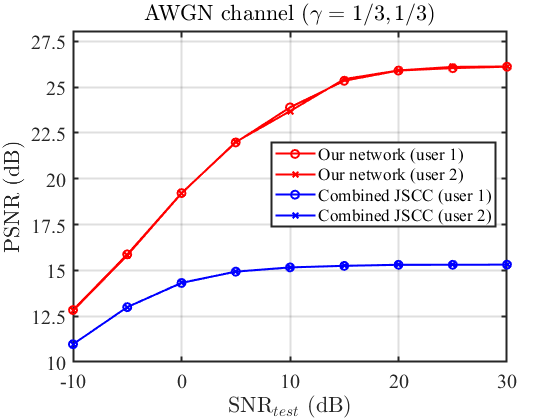}}
    \captionsetup{font={small}}
    \caption{PSNR of the two reconstructed images in our Attention-based JSCC and the combined JSCC network with one JSC encoder and JSC decoders.}
    \label{fig: PSNR 1}
\end{figure}

\subsection{Video Frame Interpolation Performance}

We evaluate the performance of the SFMA communication system for GAI-enabled video transmission, focusing on its effectiveness in video frame interpolation. This analysis underscores the significance of considering the temporal gap between simultaneously transmitted frames in the user pairing algorithm. We utilize the Multi-Scale Structural Similarity Index (MS-SSIM) metric \cite{ref45} to evaluate structural similarity and the Learned Perceptual Image Patch Similarity (LPIPS) metric \cite{ref46} to measure the perceptual similarities.

\begin{figure}[t]
    \centerline{\includegraphics[width=0.45\textwidth]{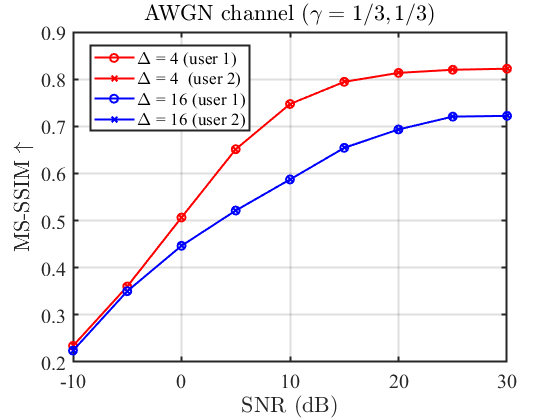}}
    \captionsetup{font={small}}
    \caption{Performance of the MS-SSIM versus different SNR levels for the scheme when $\Delta=4$ and $\Delta=16$.}
    \label{fig: MS-SSIM}
\end{figure}

MS-SSIM is a multi-scale extension of SSIM (Structural Similarity Index), and it captures multi-scale structural details, providing a comprehensive assessment of image quality. MS-SSIM is sensitive to brightness, contrast, and structure, making it suitable for evaluating image compression effects. MS-SSIM calculates the SSIM at different scales and then averages these values to obtain the final MS-SSIM value. The SSIM between the interpolated frame $\boldsymbol{\bar{x}}$ and the target frame $\boldsymbol{x}$ is calculated as
$\text{SSIM}\left( \boldsymbol{\bar{x}},\boldsymbol{x} \right) =\left[ l\left( \boldsymbol{\bar{x}},\boldsymbol{x} \right) \right] ^{\alpha}\cdot \left[ c\left( \boldsymbol{\bar{x}},\boldsymbol{x} \right) \right] ^{\beta}\cdot \left[ s\left( \boldsymbol{\bar{x}},\boldsymbol{x} \right) \right] ^{\gamma}$,
where $\alpha, \beta, \gamma$ are weighting parameters and usually set to 1, $l\left( \boldsymbol{\bar{x}},\boldsymbol{x} \right)$ denotes the brightness similarity, $c\left( \boldsymbol{\bar{x}},\boldsymbol{x} \right)$ denotes the contrast similarity, and $s\left( \boldsymbol{\bar{x}},\boldsymbol{x} \right)$ denotes the structure similarity. 
Therefore, we calculate the MS-SSIM by
\begin{equation}
    \text{MS-SSIM}\left( \boldsymbol{\bar{x}},\boldsymbol{x} \right) =\left[ l_S\left( \boldsymbol{\bar{x}},\boldsymbol{x} \right) \right] ^{\alpha _S}\cdot \prod_{j=1}^S{\left[ c_j\left( \boldsymbol{\bar{x}},\boldsymbol{x} \right) \right] ^{\beta _j}\cdot \left[ s_j\left( x,y \right) \right] ^{\gamma _j}},
\end{equation}
where $S$ is the scale number and $j$ denotes the $j$-th scale. 
However, MS-SSIM has limited capability in perceiving high-level semantic information, which is crucial for human visual system perception. Therefore, we also employed the LPIPS to measure perceptual similarity. The LPIPS is a deep learning-based metric that uses pre-trained convolutional neural networks (e.g., VGG) to extract image features and calculate differences in the feature space. LPIPS extracts the features of a frame by the pre-trained VGG network and computes the differences between the frames in the feature space. Let $\bar{f}^l, f^l$ represent the normalized feature map by layer $l$ of the VGG network with the input $\boldsymbol{\bar{x}}$ and $\boldsymbol{x}$. The LPIPS is calculated by
\begin{equation}
    \text{LPIPS}\left( \boldsymbol{\bar{x}},\boldsymbol{x} \right) =\sum_l{\frac{1}{\text{H}_l\text{W}_l}\sum_{h,w}{\lVert c_l\odot \left( \bar{f}_{hw}^{l}-f_{hw}^{l} \right) \rVert _{2}^{2}}},
\end{equation}
where $H_i, H_j$ are the height and width of the frame, the subscript $h$ and $w$ denote the $(h,w)$-th element of the feature map. The notation $\odot$ represents the scale operation.

In Fig. \ref{fig: MS-SSIM}, we show how the MS-SSIM of user 1 and user 2 changes as the SNR varies. From this figure, we see the similar performance of user 1 and user 2 when under the same temporal gap, which shows the fairness and robustness of the network. Fig. \ref{fig: MS-SSIM} shows that when SNR is small, the performance of $\Delta=4$ and $\Delta=16$ is very close. This is due to the fact that the structure of the reconstructed frames is greatly damaged by the noisy channel.  From Fig. \ref{fig: MS-SSIM}, we notice that the MS-SSIM values for both user 1 and user 2 increase consistently with higher $\text{SNR}_{test}$ levels, indicating improved structural similarity of the interpolated frames as SNR improves. From Fig. \ref{fig: MS-SSIM}, we can also observe that when $\Delta=4$, the MS-SSIM values plateau around 0.82, while when $\Delta=16$, the MS-SSIM values plateau around 0.72, demonstrating that the performance of $\Delta=4$ is $13.9\%$ better than that of $\Delta=16$. This is due to the fact that the scheme $\Delta=4$ considers a closer temporal gap than the scheme where $\Delta=16$ when grouping the different users, and the multi-scale structure of the generated interpolation frames will be more different if the temporal gap is higher. 

In Fig. \ref{fig: LPIPS}, we show how the LPIPS of user 1 and user 2 changes across various SNR levels. From this figure, we also see the similar performance of user 1 and user 2 when under the same temporal gap, which shows the fairness and robustness of the network. In this figure, we notice that when SNR increases, LPIPS values consistently decrease for both users, indicating enhanced perceptual similarity with an improved channel conditional. Fig. \ref{fig: LPIPS} shows that the overall LPIPS performance of the scheme $\Delta=4$ is better than that of the scheme $\Delta=16$. For example, the scheme $\Delta=16$ stabilizes at around 0.23. In contrast, the LPIPS of the scheme $\Delta=4$ stabilizes at approximately 0.05, which is only $21.7\%$ of the former. This result also shows the importance of considering the temporal gap when grouping the different users since it will greatly influence the visual perception of the generated interpolation frames.

\begin{figure}[t]
    \centerline{\includegraphics[width=0.45\textwidth]{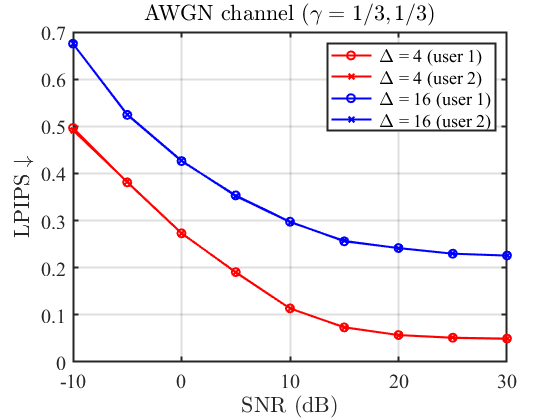}}
    \captionsetup{font={small}}
    \caption{Performance comparison of LPIPS versus different SNR levels for the scheme when $\Delta=4$ and $\Delta=16$.}
    \label{fig: LPIPS}
\end{figure}

The combined analysis of MS-SSIM and LPIPS metrics provides a comprehensive evaluation of the SFMA communication system for GAI-enabled video transmission. Both metrics exhibit a positive correlation with increasing SNR, reflecting improved video quality. Fig. \ref{fig: MS-SSIM} and Fig. \ref{fig: LPIPS} both show the alignment in trends between MS-SSIM and LPIPS across different SNRs confirming the reliability of the SFMA system in enhancing both structural and perceptual video quality. From Fig. \ref{fig: MS-SSIM} and Fig. \ref{fig: LPIPS}, the observed consistency across users further validates the robustness and fairness of the system in delivering high-quality video frames, ensuring a good user experience in semantic video transmission scenarios.

\subsection{Capability of Serving Three Users}

\begin{figure}[t]
    \centerline{\includegraphics[width=0.45\textwidth]{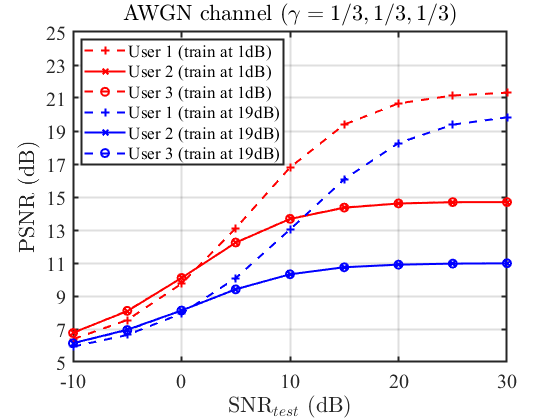}}
    \captionsetup{font={small}}
    \caption{PSNR performance of the scenario where simultaneously transmitted to 3 users.}
    \label{fig: 3 users}
\end{figure}

Fig. \ref{fig: 3 users} shows how the PSNR changes as the SNR of the testing dataset varies. This figure shows that the performance of the three-user case will decrease compared to the two-user case. This is due to the fact that the intra-group semantic interference is higher compared to the two-user case, which also indicates that when the number of simultaneously transmitted users increases, the system's transmission performance will decrease. From Fig. \ref{fig: 3 users}, we notice that the same user trained at 1dB performs better than that trained at 19dB. This indicates that the training process under worse channel conditions enables the receiver to have a higher ability to recover the damaged images. Fig. \ref{fig: 3 users} also indicates that while under the same training SNR, there will be random 2 users whose performances are the same and the remaining one has a better performance. This phenomenon shows that the robustness and fairness of the system will be decreased which will increase the complexity of designing the resource allocation strategy in the system. Therefore, we choose to pair every 2 users as a group to balance the overall performance of the system while ensuring fairness and robustness.  The proposed SFMA can also be applied to scenarios with more than two users per group. If there are more users in each group, the spectrum efficiency can be improved, while the decoding complexity can be increased since more interference is introduced. Thus, there is a trade-off between spectrum efficiency and decoding complexity when considering the number of users in each group. 

\begin{table*}[!t]
\centering
\begin{minipage}{1\textwidth}
\begin{subequations}\label{KKT condition whole}
    \begin{align}
    &\frac{\partial L}{\partial p_k} = \frac{\partial \left[ r_{k,1}\left( p_{k,1}, p_{k,2} \right)+r_{k,2}\left(p_{k,1}, p_{k,2} \right) \right]}{\ln 2 \cdot \partial p_k} + \lambda_{k,1} \left| h_{k,1} \right|^2 \left[ \left( 1 - 2^{R_{k,1}} \right) \left(\eta_{k,2} \rho_{21}^k \left( p_{k} \right) + \rho_{21}^k \left( p_{k}\right)' p_{k,2} \right) + \eta_{k,1} \right] \notag \\
    & \quad \quad+ \lambda_{k,2} \left| h_{k,2} \right|^2 \left[ \left( 1 - 2^{R_{k,2}} \right) \left( \eta_{k,1} \rho_{12}^k \left( p_{k} \right) + \rho_{12}^k \left( p_{k}\right)' p_{k,1} \right) + \eta_{k,2} \right] -  \mu = 0, \; \forall k\in \mathcal{K}, \label{eq: kkt 1}  \\
    & \lambda_{k,1} \left[ \left( p_{k,1} + \rho_{21}^k \left( p_{k} \right) p_{k,2} \right) \left| h_{k,1} \right|^2 + \sigma_{k,1}^2 - 2^{R_{k,1}} \left( \rho_{21}^k \left(p_{k} \right) p_{k,2}\left| h_{k,1} \right|^2 + \sigma_{k,1}^2 \right) \right] = 0, \; \forall k \in \mathcal{K}, \label{eq: kkt 2} \\ 
    & \lambda_{k,2} \left[ \left( p_{k,2} + \rho_{12}^k \left( p_{k}\right) p_{k,1} \right)\left| h_{k,2} \right|^2 + \sigma_{k,2}^2 - 2^{R_{k,2}} \left( \rho_{12}^k \left( p_{k}\right) p_{k,1}\left| h_{k,2} \right|^2 + \sigma_{k,2}^2 \right) \right] = 0, \; \forall k \in \mathcal{K}, \label{eq: kkt 3} \\
    & \mu \left( P_{\text{max}} - \sum_{k=1}^{K} p_k \right) = 0, \label{eq: kkt 4} \\
    &  \eqref{eq: power constraints 1}, \eqref{eq: power constraints 2}, \; \boldsymbol{\lambda}_1, \boldsymbol{\lambda}_2 \geqslant 0, \; \mu \geqslant 0.
\end{align}
\end{subequations}
\medskip
\hrule
\end{minipage}
\end{table*}

\section{Conclusion}
\label{sec: conclusion}
In this paper, we have introduced a novel multi-user multiple access semantic communication system named SFMA. To jointly optimize the semantic transmission rate and the temporal gap between the simultaneously transmitted frames, we formulate a user pairing problem. To solve this problem, we have used simulations to verify that the standard equation of signal-to-interference-plus-noise ratio cannot capture the performance of our designed multi-user multiple access semantic communication system since it cannot capture data meaning transmission performance. To address this problem, we have introduced a weight parameter into the standard SINR equation to capture the performance of our designed semantic system accurately. We then develop a user pairing algorithm to pair the two users having the highest preference value which is a weighted combination of the semantic transmission rate and the temporal gap. To further optimize the semantic transmission rates, we have formulated an optimization problem whose goal is to maximize the sum rates of all users while meeting each user's minimum rate requirement. To solve the formulated problem, we have proposed a solution that first finds the optimal power allocated to each group, and then optimizes the power allocated to each user within a group. The designed two-step solution can significantly simplify the optimization process and effectively reduce the complexity. Simulation results have shown that the proposed method yielded significant improvements in terms of positioning accuracy compared to baselines.

While the proposed SFMA system has shown promising results, several avenues remain open for further exploration and enhancement. Future research could focus on developing more sophisticated models for semantic interference that consider dynamic and adaptive features of semantic spaces, potentially leading to even more efficient resource allocation strategies. In addition, as the number of users increases, the complexity of user pairing and power allocation algorithms may become a bottleneck. Investigating scalable and computationally efficient algorithms will be essential for large-scale deployment. 

\appendix
\label{sec: appendices}

To prove Lemma 1, we first present the KKT conditions of the problem (\ref{eq: subproblem 1}) in \eqref{KKT condition whole}.

According to \eqref{KKT condition whole}, every possible value of the total power $p_k$ allocated to group $k$ should satisfy a combination value of Lagrange multipliers. Therefore, each $p_k$ can be considered as the function of $\lambda _{k,1}, \lambda _{k,2}$ and $\mu$. We further derive the three extreme points of $p_k$ in the following three situations. 
\begin{enumerate}
    \item When $\lambda _{k,1}\ne 0,\lambda _{k,2}=0$, the corresponding solution derived from \eqref{eq: kkt 2} is denoted as $p_{k_1}$, as shown in \eqref{eq: pk1}.
    \item When $\lambda _{k,1}=0,\lambda _{k,2}\ne 0$, the corresponding solution derived from \eqref{eq: kkt 3} is denoted as $p_{k_2}$, as shown in \eqref{eq: pk2}.
    \item When $\lambda _{k,1}= 0,\lambda _{k,2}=0$, the corresponding solution derived from \eqref{eq: kkt 1} is denoted as $p_{k_3}$ which is the solution of \eqref{eq: pk3}.
\end{enumerate}

\bibliographystyle{ieeetr}
\bibliography{ref}

\end{document}